\documentclass[a4paper,11pt]{article}
\usepackage[utf8x]{inputenc}
\usepackage[margin=2cm]{geometry}

\usepackage[english]{babel}

\usepackage{amsmath,amssymb,amsfonts,amsthm}
\usepackage{latexsym,mathrsfs}
\usepackage{color}
\usepackage{graphicx}
\usepackage{epstopdf}
\usepackage{caption}
\usepackage{subcaption}
\usepackage{multirow} 
\usepackage{hyperref}

\graphicspath{{plots/}}

\newtheorem{definition}{Definition}[section]
\newtheorem{lemma}[definition]{Lemma}
\newtheorem{proposition}[definition]{Proposition}
\newtheorem{theorem}[definition]{Theorem}

\definecolor{detailsgray}{gray}{0.3}

\def \cbb{\mathbb{C}}

\def \rbb{\mathbb{R}}

\def \dcal {\mathcal{D}}

\def \hcal {\mathcal{H}}

\def \scal {\mathcal{S}}

\newcommand{\betav}{\boldsymbol{\beta}}
\newcommand{\thetav}{\boldsymbol{\theta}}
\newcommand{\etav}{\boldsymbol{\eta}}

\newcommand{\piv}{\boldsymbol{\pi}}

\newcommand{\varphiv}{\boldsymbol{\varphi}}

\newcommand{\zerov}{\boldsymbol{0}}

\def \. { \,\! }

\def\clap#1{\hbox to 0pt{\hss#1\hss}}

\def \cdotarg { \, \cdot \, }

\DeclareMathOperator{\im}{Im}
\DeclareMathOperator{\re}{Re}

\newcommand{\idop}{\boldsymbol{1}}

\def \noqed{ \renewcommand{\qed}{} }

\newcommand{\hscalar}[2]{\langle #1 , #2 \rangle }

\newcommand{\zd}{z^\dagger}

\newcommand{\Fmin}{F_{\mathrm{min}}}

\newcommand{\opt}{T_\mathrm{one}^{00}(g^2)}
\numberwithin{equation}{section}

\title{Negative energy densities in integrable quantum field theories at one-particle level}
\author{Henning Bostelmann\thanks{University of York, Department of Mathematics, York YO10 5DD, United Kingdom;
 e-mail: \mbox{\tt henning.bostelmann@york.ac.uk}} 
\and
Daniela Cadamuro\thanks{University of Bristol, School of Mathematics, University Walk, Bristol BS8 1TW,
United Kingdom.
Present address: Mathematisches Institut, Universit\"at G\"ottingen, Bunsenstra\ss{}e 3-5, 37073 G\"ottingen, Germany;
e-mail: \mbox{\tt daniela.cadamuro@mathematik.uni-goettingen.de} }
}

\date{2 March 2016}

\begin{document}

\maketitle

\begin{abstract}
We study the phenomenon of negative energy densities in quantum field theories with self-interaction. Specifically, we consider a class of integrable models (including the sinh-Gordon model) in which we investigate the expectation value of the energy density in one-particle states. In this situation, we classify the possible form of the stress-energy tensor from first principles. We show that one-particle states with negative energy density generically exist in non-free situations, and we establish lower bounds for the energy density (quantum energy inequalities). Demanding that these inequalities hold reduces the ambiguity in the stress-energy tensor, in some situations fixing it uniquely.
Numerical results for the lowest spectral value of the energy density allow us to demonstrate how negative energy densities depend on the coupling constant and on other model parameters.
\end{abstract}

\section{Introduction} \label{sec:intro}

The energy density is one of the fundamental observables in classical as well as quantum field theories. It has a special significance in field theories on curved backgrounds, since it enters Einstein's field equation as a source term, and is therefore linked to the structure of space-time. But also on flat Minkowski space, as well as in low-dimensional conformal field theories, it plays an important role.

In the transition from classical to quantum field theories, some of the distinctive properties of the energy density are lost. In particular, the classical energy density is positive at every point, which in General Relativity implies certain stability results, such as the absence of wormholes \cite{FordRoman:1996}. This is however not the case in quantum field theory, not even on flat spacetime: While the global energy operator $H$ is still nonnegative, the energy density can have arbitrarily negative expectation values \cite{EGJ}. However, a remnant of positivity is still expected to hold. When one considers local averages of the energy density, $T^{00}(g^2)=\int dt \, g^2(t)\, T^{00}(t,0)$ for some fixed smooth real-valued function $g$, then certain lower bounds -- quantum energy inequalities (QEIs) -- should be satisfied. In the simplest case, one finds for any given averaging function $g$ a constant $c_g > 0$ such that
\begin{equation}\label{eq:qei0}
  \hscalar{\varphi}{T^{00}(g^2) \varphi} \geq - c_g \lVert \varphi \rVert^2
\end{equation}
for all (suitably regular) vector states $\varphi$ of the system; this is a so-called \emph{state-independent QEI}. In general, \eqref{eq:qei0} may need to be replaced with a somewhat weaker version (a state-\emph{dependent} QEI) where the right-hand side can include a slight dependence on the total energy of the state $\varphi$.

This raises the question under which conditions the QEI \eqref{eq:qei0} can be shown to hold rigorously. The inequality has in fact been established for various \emph{linear} quantum fields, on flat as well as curved spacetime,
 and in conformal QFTs 
(e.g., \cite{Flanagan:1997,PfenningFord_static:1998,Fews00,FewsterVerch_Dirac,FewsterHollands:2005}; see
 \cite{Fewster:lecturenotes} for a review). 
However, dropping the restriction to linear fields, that is, allowing for \emph{self-interacting} quantum field theories, few results are available. This is not least due to the limited availability of rigorously constructed quantum field theoretical models; see however our recent proof of \eqref{eq:qei0} in the massive Ising model \cite{BostelmannCadamuroFewster:ising}. In a model-independent setting, one can establish state-dependent inequalities for certain ``classically positive'' expressions \cite{BostelmannFewster:2009}, based on a model-independent version of the operator product expansion \cite{Bos:product_expansions}, but the relation of these expressions to the energy density remains unknown.

In the present paper, we will investigate the inequality \eqref{eq:qei0} in a specific class of self-interacting models on 1+1 dimensional Minkowski space, so-called \emph{quantum integrable models}, which have recently become amenable to a rigorous construction. Specifically, we consider integrable models with one species of massive scalar boson and without bound states.

An a priori question is what form the stress-energy tensor $T^{\alpha\beta}$ takes in these models. There is a straightforward answer in models derived from a classical Lagrangian, such as the sinh-Gordon model, where a candidate for the operator can be computed \cite{FringMussardoSimonetti:1993,KoubekMussardo:1993,MussardoSimonetti:1994}. However, we also consider theories where no associated Lagrangian is known; and more generally, we aim at an intrinsic characterization of the quantum theory without referring to a ``quantization process''. In fact, there will usually be more than one local field that is compatible with generic requirements on $T^{\alpha\beta}$, such as covariance, the continuity equation, and its relation to the global Hamiltonian. 

Given $T^{\alpha\beta}$, one can ask whether the QEI \eqref{eq:qei0} holds for this stress-energy tensor, or rather \emph{for which choice} of stress-energy tensor. In fact, QEIs may hold for some choices of $T^{\alpha\beta}$ but not for others, as the nonminimally coupled free field on Minkowski space shows \cite{FewsterOsterbrink2008}. Ideally, one would hope that requiring a QEI fixes the stress-energy tensor uniquely.

In the present article, we will consider the above questions in integrable models, but with one important restriction: We will consider the energy density \emph{at one-particle level} only. That is, we will ask whether the inequality \eqref{eq:qei0} holds for all (sufficiently regular) \emph{one-particle} states $\varphi$. 

This case might seem uninteresting at first: One might argue that in integrable models, where the particle number is a conserved quantity, the effect of interaction between particles is absent at one-particle level. But this is \emph{not} the case, as already the massive Ising model shows \cite{BostelmannCadamuroFewster:ising}: in this model of scalar bosons, one-particle states with negative energy density exist, whereas in a model of free bosons, the energy density is positive at one-particle level. We will demonstrate in this paper that self-interaction in our class of models leads to negative energy density in one-particle states, and that this effect increases with the strength of the interaction.

Our approach is as follows. Having recalled the necessary details of the integrable models considered (Sec.~\ref{sec:integrable}), we ask what form the stress-energy tensor can take at one-particle level. This will lead to a full characterization of the integral operators involved, with the energy density fixed up to a certain polynomial expression (Sec.~\ref{sec:edensity}). 

In Sec.~\ref{sec:negative}, we will show that under very generic assumptions, states of negative energy density exist, and that for certain choices of the stress-energy tensor, the energy density becomes so negative that the QEI \eqref{eq:qei0} \emph{cannot} hold even in one-particle states. For other choices of the energy density operator, we demonstrate in Sec.~\ref{sec:qei} that the QEI \emph{does} hold. In consequence (Sec.~\ref{sec:models}), we find in the massive Ising model that one-particle QEIs hold for \emph{exactly one} choice of energy density, whereas in other models (including the sinh-Gordon model), the choice of energy density is at least very much restricted by a QEI. 

All this is based on rigorous estimates for the expectation values of $T^{00}(g^2)$. However, the best possible constant $c_g$ in \eqref{eq:qei0} -- in other words, the lowest spectral value of $T^{00}(g^2)$ restricted to one-particle matrix elements -- can only be obtained by numerical approximation. We discuss the results of an approximation scheme in Sec.~\ref{sec:numerical}, thus demonstrating how the effect of negative energy density varies with the coupling constant and with the form of the scattering function. The program code used for this purpose is supplied with the article \cite{suppl:code}. We end with a brief outlook in Sec.~\ref{sec:conclusions}.

\section{Integrable models}\label{sec:integrable}

\begin{table}
 \begin{tabular}{l|l|l|l}
     model & $S(\zeta)$ & $\Fmin(\zeta)$ & parameters \\
     \hline
     free field & $\hphantom{-}1$ & $1$ \\[10pt]
     Ising  & $ -1$ & $-i\sinh\frac{\zeta}{2}$\\[10pt]
     sinh-Gordon & $\hphantom{-}\dfrac{\sinh \zeta - i \sin B \pi/2 }{\sinh \zeta + i \sin B \pi/2 }$ 
     & $\exp J_B(\zeta)$ & $0 < B < 2$ \\[15pt]
     \parbox{2.5cm}{(generalized sinh-Gordon)} & $\hphantom{-}\displaystyle\prod_{j =1}^n \dfrac{\sinh \zeta - i \sin B_j \pi/2 }{\sinh \zeta + i \sin B_j \pi/2 }$ 
     & $\dfrac{\exp \textstyle\sum_{j =1}^n J_{B_j}(\zeta)}{(-i \sinh \tfrac{\zeta}{2})^{2 \lfloor n/2 \rfloor}}$ &
     \multirow{2}{*}{
      \parbox{3.5cm}{\raggedright $B_j \in (0,2) + i \rbb$, either real or in complex-conjugate pairs}
     }
     \\[15pt]
     (generalized Ising) & $-\displaystyle\prod_{j =1}^n \dfrac{\sinh \zeta - i \sin B_j \pi/2 }{\sinh \zeta + i \sin B_j \pi/2 }$ 
     & $\dfrac{ \exp \textstyle\sum_{j =1}^n J_{B_j}(\zeta)}{(-i \sinh \tfrac{\zeta}{2})^{2 \lceil n/2 \rceil-1}}$ \\[10pt]
 \end{tabular}
 \caption{Examples of scattering functions $S$ and corresponding minimal solutions $\Fmin$. See Eq.~\eqref{eq:jdef} for the definition of the function $J_B$.}
 \label{tab:models}
\end{table}

For our investigation, we will use a specific class of quantum field theoretical models on 1+1 dimensional spacetime, a simple case of so-called \emph{integrable models} of quantum field theory. These models describe a single species of scalar massive Bosons with nontrivial scattering. The scattering matrix is \emph{factorizing}: When two particles with rapidities $\theta$ and $\eta$ scatter, they exchange a phase factor $S(\theta-\eta)$; and multi-particle scattering processes can be described by a sequence of two-particle scattering processes. 

There are several approaches to constructing such integrable quantum field theories. Conventionally, one starts from a classical Lagrangian, derives the two-particle scattering function $S$ from there, and then constructs local operators (quantum fields) by their matrix elements in asymptotic scattering states; this is the \emph{form factor programme} \cite{Smirnov:1992, BabujianFoersterKarowski:2006}. A more recent, alternative approach \cite{SchroerWiesbrock:2000-1, Lechner:2008} starts from the function $S$ as its input, then constructs quantum field localized in spacelike wedges (rather than at spacetime points), and uses these to abstractly obtain observables localized in bounded regions. 

We will largely follow the second mentioned approach here; in particular, we set out from a function $S$ rather than from a classical Lagrangian. In all what follows, we will assume that a scattering function $S$ is given, which we take to be a meromorphic function on $\cbb$ which fulfils the symmetry properties
\begin{equation}\label{eq:sprop}
   S(-\zeta) = S(\zeta)^{-1} = S(\zeta + i \pi) = \overline{S(\bar \zeta)}.
\end{equation}
A range of examples for such functions $S$ can easily be given, in particular because the properties \eqref{eq:sprop} are preserved under taking products of functions; see Table~\ref{tab:models}. This includes the sinh-Gordon model, depending on a coupling parameter $B$, which is normally constructed from a Lagrangian \cite{FringMussardoSimonetti:1993}; but for other examples (e.g., the generalized sinh-Gordon models mentioned in  Table~\ref{tab:models}), no corresponding Lagrangian is known.

We will not enter details of the construction of the associated quantum field theory based on $S$ here, but will recall only the general concepts as far as relevant to the present analysis. The single particle space of the theory is given by $\hcal_1=L^2(\rbb,d\theta)$, where the variable $\theta$ of the wave function is rapidity, linked to particle two-momentum $p$ by $p(\theta)=\mu(\cosh \theta,\sinh\theta)$; here $\mu>0$ is the particle mass. On $\hcal_1$, the usual representation of the Poincar\'e group acts. One then constructs an ``$S$-symmetric'' Fock space $\hcal$ over $\hcal_1$, on which ``interacting'' annihilation and creation operators $z(\theta)$ and $\zd(\theta)$ act; instead of the CCR, they fulfil the Zamolodchikov-Faddeev relations \cite{Lechner:2008}, depending on $S$. The quantum field theory is constructed on this Fock space. If $A$ is an operator of the theory (of a certain regularity class, including smeared Wightman fields),  localized in a bounded spacetime region, then it can be written in a series expansion \cite{BostelmannCadamuro:expansion,BostelmannCadamuro:characterization}
\begin{equation}\label{eq:expansion}
A = \sum_{m,n=0}^\infty \int \frac{d\thetav \, d\etav}{m!n!} F_{m+n}^{[A]}(\thetav+i\zerov,\etav+i\piv-i\zerov) z^{\dagger}(\theta_1)\cdots z^\dagger(\theta_m) z(\eta_1)\cdots z(\eta_n),
\end{equation}
where  $F_{m+n}$ are meromorphic functions with certain analyticity, symmetry and growth properties (which we will recall where we need them). Examples of such local observables $A$ would include smeared versions of the energy density, supposing they fall into the regularity class mentioned.

In the construction of functions $F_k$ that fulfil these properties, an important ingredient is the so-called \emph{minimal solution} $\Fmin$ of the model \cite{KarowskiWeisz:1978}. We consider it here with the following conventions.

\begin{definition}\label{def:fmin}
  Given a scattering function $S$, a \emph{minimal solution} is a meromorphic function $\Fmin$ on $\cbb$ which  has the following properties.
\begin{enumerate} 
\renewcommand{\theenumi}{(\alph{enumi})}
\renewcommand{\labelenumi}{\theenumi}
\item \label{it:FminNoZeros}
$\Fmin$ has neither poles nor zeros in the strip $0 \leq  \im\zeta \leq \pi$, except for a first-order zero at $\zeta=0$ in the case that $S(0)=-1$,
\item \label{it:FminIPiMinus}
  $ \Fmin(i\pi +\zeta) = \Fmin(i\pi - \zeta)$;
\item \label{it:FminMinus}
  $\Fmin(-\zeta) = S(-\zeta)\Fmin(\zeta)$;
\item \label{it:FminNorm}
  $\Fmin(i\pi) = 1$;
\item \label{it:FminBounds}
There are constants $a,b>0$ such that
$\big\lvert \log \lvert \Fmin(\zeta) \rvert \big\rvert \leq a \lvert \re \zeta \rvert + b $
 if $\lvert \re \zeta \rvert \geq 1$, $0 \leq \im \zeta \leq \pi$. 
\end{enumerate}
\end{definition}

Note that properties \ref{it:FminNoZeros} and \ref{it:FminBounds} automatically hold analogously for the strip $\rbb + i [\pi,2\pi]$ by property \ref{it:FminIPiMinus}. The first-order zero at $\zeta = 0$ (and analogously $\zeta = 2\pi$) must necessarily occur in the case $S(0)=-1$ due to \ref{it:FminMinus}.

The properties \ref{it:FminNoZeros}--\ref{it:FminBounds} actually fix $\Fmin$ uniquely if it exists, so that we can speak of \emph{the} minimal solution. We prove this in our context; cf.~\cite[p.~459]{KarowskiWeisz:1978}.

\begin{lemma}\label{lem:fminunique}
 For given $S$ fulfilling \eqref{eq:sprop}, there exists at most one minimal solution $\Fmin$.
\end{lemma}

\begin{proof}
 Given two minimal solutions $\Fmin^A$, $\Fmin^B$, define $G(\zeta) := \Fmin^A(\zeta)/\Fmin^B(\zeta)$. By property \ref{it:FminNoZeros}, this function is analytic in a neighbourhood of the strip $\rbb + i [0,2\pi]$ -- the possible zeros of $\Fmin^{A,B}$ at $\zeta = 0$ and $\zeta=2\pi i$ cancel -- and by \ref{it:FminIPiMinus} and \ref{it:FminMinus}, we have 
 \begin{equation}\label{eq:gsymm}
G(\zeta+2\pi i) = G(\zeta) =G(-\zeta).
 \end{equation}
 We can therefore find an entire function $P$ such that $G(\zeta)=P(\cosh\zeta)$: We observe 
that $\cosh(\cdotarg)$ is bijective from the region $(0,\infty) \pm i (0,\pi)$ to the upper, respectively lower, half-plane, and use the properties \eqref{eq:gsymm} to accommodate the branch cuts of the inverse hyperbolic function. From property \ref{it:FminNorm}, this function fulfils the estimate
\begin{equation}
   \log |P(z)| \leq a' \lvert \re \operatorname{arcosh} z \rvert + b'
\end{equation}
with certain constants $a',b'>0$ and for large $\re \operatorname{arcosh} z $. Since $\re \operatorname{arcosh} z $ grows like $\log 2|z|$ for large $|z|$, this means that $P$ is polynomially bounded at infinity, and hence a polynomial. But due to property \ref{it:FminNoZeros}, $P$ has no zeros, and is therefore constant. Now from \ref{it:FminNorm}, $P(z)=P(-1)=1$. Hence $\Fmin^A=\Fmin^B$.
\end{proof}

Let us note a simple consequence: One checks that together with $\Fmin$, also $\zeta\mapsto\overline{\Fmin(-\bar\zeta)}$ fulfils properties \ref{it:FminNoZeros}--\ref{it:FminBounds}. Thus the lemma yields $\Fmin(-\bar\zeta)=\overline{\Fmin(\zeta)}$. Together with \ref{it:FminIPiMinus}, this shows that $\Fmin$ is symmetric and real-valued on the line $\rbb + i \pi$. We will use this fact frequently in the following.

In this article, we will always assume that a minimal solution $\Fmin$ exists. In fact, for the examples we mentioned, they are listed in Table~\ref{tab:models}. For the sinh-Gordon and related models, this involves the integral expression
\begin{equation}\label{eq:jdef}
 J_B(\zeta+i\pi) := 8\int_{0}^{\infty} \frac{dx}{x}\frac{\sinh\frac{x B}{4}\sinh\frac{x(2-B)}{4} \sinh\frac{x}{2}}{\sinh^2 x} 
\sin^2 \frac{x\zeta}{2\pi} \qquad \text{for } B \in (0,2)+i\rbb
\end{equation}
which is known from \cite{FringMussardoSimonetti:1993} (but note that we use a different normalization for $\Fmin$).
Due to the Riemann-Lebesgue lemma, $J_B(\zeta)$ converges to a constant as $\lvert\re\zeta\rvert \to \infty$, with $0 < \im \zeta < 2 \pi$ fixed.

For the generalized sinh-Gordon and Ising models in Table~\ref{tab:models}, $\Fmin$ can essentially be obtained as a product of the minimal solutions
of the corresponding sinh-Gordon or Ising factors, since properties \ref{it:FminIPiMinus}--\ref{it:FminBounds} in Def.~\ref{def:fmin} are again preserved under products.
However, in order to satisfy property \ref{it:FminNoZeros}, any possible double zeros of the product function at $\zeta=0$ need to be cancelled by dividing by appropriate powers of $(-i\sinh (\zeta/2))^2$.

\section{Energy density at one-particle level}\label{sec:edensity}

The first question we want to consider is which form the energy density operator can take in our models. More specifically, we ask 
what the functions $F_k[A]$ in the expansion \eqref{eq:expansion} can be if 
\begin{equation}
A =  T^{\alpha\beta}(f) = \int dt\, f(t) T^{\alpha\beta}(t,0)
\end{equation}
is a component of the stress-energy tensor smeared with a real-valued test function in time direction.  
The answer may appear obvious in models such as the free field or the sinh-Gordon model, where the energy density is linked to the classical Lagrangian and well studied. However, we  aim at an intrinsic characterization of the energy density within the quantum theory, and therefore we are looking for the most general form of the stress-energy tensor compatible with generic assumptions on this operator, which will be detailed below.

As announced in the introduction, we will consider only one-particle states of the theory, and evaluate the stress-energy tensor only in these. More precisely, we will consider the stress-energy tensor only in matrix elements of the form
\begin{equation}
  \hscalar{\varphi}{T^{\alpha\beta}(f) \, \psi} \quad \text{with } \varphi,\psi \in \dcal(\rbb) \subset \hcal_1.
\end{equation}
(The restriction of the quadratic form to smooth functions of compact support, i.e., $\varphi,\psi\in\dcal(\rbb)$, is perhaps too cautious -- the form can easily be extended to non-smooth and to sufficiently rapidly decaying wave functions, and we will in fact use piecewise continuous functions in the numeric evaluation in Sec.~\ref{sec:numerical}; but for the moment we restrict to $\dcal(\rbb)$ for simplicity.)

Writing $T^{\alpha\beta}(f)$ in expanded form as in \eqref{eq:expansion}, we see that only the coefficients $F_0$ and $F_2$ contribute to this one-particle matrix element. Since the coefficient $F_0[A]$ equals the vacuum expectation value of the operator, we can assume without loss that $F_0[T^{\alpha\beta}(f)]=0$; in fact, this is necessary for the energy density, since it would otherwise not integrate to the Hamiltonian $H$. This leaves us with
\begin{equation}
  \hscalar{\varphi}{T^{\alpha\beta}(f)\,\psi} = \int d\theta \, d\eta \, \overline{\varphi(\theta)} F_2[T^{\alpha\beta}(f)] (\theta+i0,\eta+i\pi-i0) \psi(\eta).
\end{equation}
Since we expect $T^{\alpha\beta}(f)$ to be a translation-covariant operator-valued distribution in $f$, we will assume that
\begin{equation}\label{eq:F2smearing}
 F_2[T^{\alpha\beta}(f)] (\theta,\eta+i\pi) = F^{\alpha\beta} (\theta,\eta)  \tilde f \big( \mu\cosh\theta - \mu\cosh\eta \big)
\end{equation}
with a function $F^{\alpha\beta}$ independent of $f$, where we take the Fourier transform with the convention $\tilde f (p) = \int dt\,f(t) e^{itp}$. 

The task is therefore to determine the possible form of $F^{\alpha\beta}$, starting from physical properties.
We will list these assumptions one by one and explain their motivation, but we skip details of how they are derived; that is, we will take these assumptions as axioms in our context.

The first set of conditions follows from the general properties of the expansion coefficients $F_k{[A]}$ of a local operator $A$ in integrable models, as derived in \cite{BostelmannCadamuro:characterization}.

\begin{enumerate}
\renewcommand{\theenumi}{(T\arabic{enumi})}
\renewcommand{\labelenumi}{\theenumi}

\item \label{it:Fab-analytic} 
$F^{\alpha\beta}$ are meromorphic functions on $\cbb^2$, analytic in a neighbourhood of the region $-\pi \leq \im(\theta-\eta)\leq \pi$.

This is due to the general analyticity properties of the coefficients $F_k$ \cite[property (FD1)]{BostelmannCadamuro:characterization}, together with the absence of ``kinematic poles'' in the specific case of the coefficient $F_2$ (from property (FD4) there).
 
\item \label{it:Fab-formfact} 
They have the symmetry properties
\begin{equation}
  F^{\alpha\beta}(\theta,\eta) 
  = S(\theta-\eta) F^{\alpha\beta}(\eta+i\pi,\theta-i\pi)
  = F^{\alpha\beta}(\eta+i\pi,\theta+i\pi).
\end{equation}

This is a rewritten form of the properties of ``$S$-symmetry'' and ``$S$-periodicity'' -- properties (FD2) and (FD3) in \cite{BostelmannCadamuro:characterization}.

\item\label{it:Fab-hermitean} 
Hermiticity of the observable $T^{\alpha\beta}$ is expressed as  
\begin{equation}
  F^{\alpha\beta}(\theta,\eta) = \overline{F^{\alpha\beta}(\eta,\theta)}	\quad \text{for all }\theta,\eta\in\rbb.
\end{equation} 

\item \label{it:Fab-growth}
We demand that there exist constants $k,\ell>0$ such that
\begin{equation}\label{eq:Fab-growth}
 \lvert F^{\alpha\beta}(\theta,\eta) \rvert \leq  \ell (\cosh \re \theta)^k (\cosh \re \eta)^k
 \quad
 \text{ whenever } 
 -\pi < \im (\theta-\eta) < \pi. 
\end{equation}

This is motivated by the bounds for $F_k{[A]}$ discussed in \cite{BostelmannCadamuro:characterization}; the condition guarantees that the smeared version $T^{\alpha\beta}(f)$ for every Schwartz function $f$ will fulfil, at one-particle-level, polynomial high-energy bounds in the form of  \cite[property (FD6)]{BostelmannCadamuro:characterization}.
\end{enumerate}

\noindent
Further conditions are derived from properties that one expects specifically of the stress-energy tensor $T^{\alpha\beta}$.
\begin{enumerate}
\renewcommand{\theenumi}{(T\arabic{enumi})}
\renewcommand{\labelenumi}{\theenumi}
\setcounter{enumi}{4}

\item\label{it:Fab-tenssymm}
 The stress-energy tensor is a symmetric tensor, which rewrites in our context as 
\begin{equation}\label{eq:Fab-tenssymm}
  F^{\alpha\beta}(\theta,\eta) = F^{\beta\alpha}(\theta,\eta).
\end{equation}

\item\label{it:Fab-covar} 
It is covariant under Lorentz transformations, which means in our terms, cf.~\cite[Prop.~3.9]{BostelmannCadamuro:expansion},
\begin{equation}\label{eq:Fab-covar}
  F^{\alpha\beta}(\theta-\lambda,\eta-\lambda) = \Lambda^\alpha_{\alpha'} \Lambda^\beta_{\beta'} F^{\alpha'\beta'}(\theta,\eta),
\end{equation}
 where $\Lambda$ is the boost matrix with rapidity parameter $\lambda$.

 \item\label{it:Fab-reflect} 
It is invariant under spacetime reflections, which by \cite[Thm.~5.4]{BostelmannCadamuro:characterization} translates to
\begin{equation}\label{eq:Fab-reflect}
  F^{\alpha\beta}(\theta+i\pi,\eta+i\pi) = F^{\alpha\beta}(\theta,\eta).
\end{equation}

\item \label{it:Fab-contin}
It fulfils the continuity equation ($\partial_\alpha T^{\alpha\beta} = 0$), which reads in our terms,
\begin{equation}\label{eq:Fab-contin}
  \big(p_\alpha(\theta)-p_\alpha(\eta)\big) F^{\alpha\beta}(\theta,\eta) = 0,
\end{equation}
where $p(\theta)=\mu(\cosh\theta,\sinh\theta)$.
\item \label{it:Fab-energy}
The (0,0)-component of the tensor integrates to the Hamiltonian, $\int dx\, T^{00}(t,x) = H$, which translates to
\begin{equation}\label{eq:Fab-energy}
  F^{00}(\theta,\theta) = \frac{\mu^2}{2\pi} \cosh^2 \theta.
\end{equation}
\end{enumerate}

Taking these properties as our starting point, we ask what functions $F^{\alpha\beta}$ are compatible with them. 
The answer is given in the following proposition.

\begin{proposition}\label{prop:stressenergy}
 Functions $F^{\alpha\beta}$ fulfil the properties \ref{it:Fab-analytic}--\ref{it:Fab-energy} if, and only if, there
 exists a real polynomial $P$ with $P(1)=1$ such that
 \begin{equation}\label{eq:stressenergych}
   F^{\alpha\beta}(\theta,\eta) = F^{\alpha\beta}_\mathrm{free}(\theta,\eta) P(\cosh(\theta-\eta)) \Fmin(\theta-\eta+i\pi),
 \end{equation}
 where
 \begin{equation}
  F^{\alpha\beta}_\mathrm{free}(\theta,\eta) = \frac{\mu^2}{2\pi} \begin{pmatrix}
                                 \cosh^2 (\frac{\theta+\eta}{2}) & \frac{1}{2} \sinh(\theta+\eta) \\
                                 \frac{1}{2} \sinh(\theta+\eta) & \sinh^2 (\frac{\theta+\eta}{2})
                              \end{pmatrix}.
 \end{equation}
\end{proposition}

Note that $F^{\alpha\beta}_\mathrm{free}$ is the well-known one-particle expression of the ``canonical'' stress-energy tensor of the free Bose field.

\begin{proof}
It is straightforward to check that $F^{\alpha\beta}$ as given in \eqref{eq:stressenergych} fulfils all conditions \ref{it:Fab-analytic}--\ref{it:Fab-energy}, knowing that $F^{\alpha\beta}_\mathrm{free}$ fulfils them in the case $S=1$.

Thus, let $F^{\alpha\beta}$ fulfil conditions \ref{it:Fab-analytic}--\ref{it:Fab-energy}. 
We first use \ref{it:Fab-contin} with $\beta = 0$ and with $\beta=1$, along with \ref{it:Fab-tenssymm}, to obtain
\begin{equation}\label{eq:F11F00}
  F^{11}(\theta,\eta) = \tanh^2 \frac{\theta+\eta}{2} F^{00}(\theta,\eta).
\end{equation}
Now we consider the functions $G^{\alpha\beta}(\zeta):=F^{\alpha\beta}(\zeta/2,-\zeta/2)$, which are meromorphic by \ref{it:Fab-analytic}. Note that \ref{it:Fab-contin} with $\beta=0$ implies
\begin{equation}
 \underbrace{\big(p_0(\zeta)-p_0(-\zeta)\big)}_{=0} G^{00}(2\zeta) 
 +\underbrace{\big(p_1(\zeta)-p_1(-\zeta)\big)}_{=-2 \mu \sinh \zeta} G^{10}(2\zeta) = 0
\end{equation}
and hence $G^{10}=0$. Then $G^{01}=0$ by \ref{it:Fab-tenssymm}. Also, \eqref{eq:F11F00} leads to $G^{11}=0$, and the only nonzero component of $G$ is $G^{00}$. From \ref{it:Fab-formfact} and \ref{it:Fab-reflect} one concludes
\begin{equation}\label{eq:g00symm}
   G^{00}(-\zeta) = G^{00}(\zeta), \quad
   G^{00}(\zeta-i\pi) = S(\zeta) \, G^{00}(\zeta+i\pi).   
\end{equation}
Now set
\begin{equation}
  Q(\zeta) := \frac{2\pi}{\mu^2} G^{00}(\zeta)/\Fmin(\zeta+i\pi).
\end{equation}
This $Q$ is analytic in a neighbourhood of the strip $\rbb+i[-\pi,\pi]$. (Note that in the case $S(0)=-1$, the zeros of the denominator at $\zeta = \pm i \pi$ are cancelled by corresponding zeros of the numerator which exist due to \eqref{eq:g00symm}.)
The symmetry relations \eqref{eq:g00symm} and Def.~\ref{def:fmin}\ref{it:FminIPiMinus},\ref{it:FminMinus} imply that
\begin{equation}\label{eq:qsymm}
 Q(-\zeta) = Q(\zeta), \quad
 Q(\zeta + i \pi  ) = Q(\zeta - i \pi).
\end{equation}
Arguing as in the proof of Lemma~\ref{lem:fminunique}, we can therefore find an entire function $P$ such that $Q(\zeta)=P(\cosh\zeta)$.
From property \ref{it:Fab-growth} and Def.~\ref{def:fmin}\ref{it:FminBounds}, this $P$ fulfils the estimates
\begin{equation}
  \lvert  P(z) \rvert 
  \leq \frac{2\pi}{\mu^2} \frac{\big\lvert G^{00}(\operatorname{arcosh} z) \big\rvert }{\lvert \Fmin(\operatorname{arcosh} (z) + i \pi) \rvert}   \leq \ell' \frac{\big( \cosh \tfrac{1}{2} \re \operatorname{arcosh} z \big)^{2k}}{ (\cosh \re \operatorname{arcosh} z)^{k'} }
  \leq \ell' \big( 2 |z| + 1)^{k-k'}
\end{equation}
with some constants $\ell', k'$. That is, $P$ is a polynomially bounded entire function, and hence a polynomial. From \ref{it:Fab-hermitean}, we can conclude that $P(\bar z)=\overline{P(z)}$ for all $z$, thus the coefficients of $P$ are real. Also, \ref{it:Fab-energy} implies $G^{00}(0)=\mu^2/2\pi$ and hence $P(1)=1$. Thus $P$ has the properties claimed in the proposition. Combining our results, we have shown that
\begin{equation}
  F^{\alpha\beta}\big(\tfrac{\zeta}{2},-\tfrac{\zeta}{2}\big) = \frac{\mu^2}{2\pi} \begin{pmatrix} 1 & 0 \\ 0 & 0 \end{pmatrix} P(\cosh \zeta) \Fmin(\zeta+i\pi).
\end{equation}
This means that \eqref{eq:stressenergych} holds in the case $\theta+\eta=0$. But then it holds for any value of $\theta+\eta$, since both sides of \eqref{eq:stressenergych} fulfil the covariance condition \ref{it:Fab-covar}.
\end{proof}

In the following, when we speak of the stress-energy tensor of a model, we will always refer to one of the form \eqref{eq:F2smearing}, with $F^{\alpha\beta}$ as in Proposition~\ref{prop:stressenergy}. We will often abbreviate
\begin{equation}
  F_P(\zeta) := P(\cosh \zeta) \Fmin(\zeta+i\pi),
\end{equation}
noting that $F_P$ enjoys most of the defining properties of $\Fmin$ (namely, Def.~\ref{def:fmin}\ref{it:FminIPiMinus}--\ref{it:FminBounds} with shifted argument, but not \ref{it:FminNoZeros}). Also, $F_P$ is symmetric and real-valued on the real line.
The expectation value of the energy density now becomes
\begin{equation}\label{eq:t00f-expect}
  \hscalar{\varphi}{T^{00}(f)\,\varphi} = \frac{\mu^2}{2\pi} \int d\theta \, d\eta \, \overline{\varphi(\theta)} \varphi(\eta)
 \cosh^2 \frac{\theta+\eta}{2} \, F_P(\theta-\eta)\,
  \tilde f(\mu\cosh\theta-\mu\cosh\eta) .
\end{equation}

Proposition~\ref{prop:stressenergy} shows that on the one-particle level, we recover the well-known ``canonical form'' of the energy density of the free field and of the sinh-Gordon model \cite{FringMussardoSimonetti:1993}, and the $T^{00}$ considered for the massive Ising model in \cite{BostelmannCadamuroFewster:ising}, up to a possible polynomial factor $P$ in $\cosh(\theta-\eta)$. We emphasize that, while one might expect $P=1$, all our assumptions so far are perfectly compatible with a more general polynomial $P$. However, we will see later (in Sec.~\ref{sec:models}) that the choice of $P$ is restricted, in some cases uniquely to $P=1$, if we demand that quantum energy inequalities hold.  

\section{States with negative energy density}\label{sec:negative}

As the next question about properties of the energy density, we will ask whether single-particle states with negative energy density exist at all; more specifically, whether $ \hscalar{\varphi}{T^{00}(g^2)\varphi}$ can be negative if $\varphi\in\dcal(\rbb)$ and $g$ is a real-valued Schwartz function, i.e., $g \in\scal_\rbb(\rbb)$. The example of the free field with canonical energy density ($S=1$, $F_P=1$) shows that this is not guaranteed: In this specific case, $T^{00}(g^2)$ is known to be positive between one-particle states. However, as we shall see in a moment, the introduction of interaction quite generically leads to negative energy densities.

We will exhibit these negative energy densities by explicitly constructing corresponding states $\varphi$. In preparation, we fix a nonnegative, smooth, even function $\chi$ with support in $[-1,1]$. For $q\geq 1$, $\rho>0$, we set $\chi_{q,\rho} (\theta) := \rho^{-1/q} \lVert \chi \rVert_q^{-1} \chi(\theta/\rho)$, so that $\chi_{q,\rho}$ has support in $[-\rho,\rho]$ and is normalized with respect to the $L^q$ norm $ \lVert \cdotarg \rVert_q$.

\begin{proposition}\label{prop:negative}
Suppose that there is $\theta_P \in \rbb$ such that $\lvert F_P(\theta_P) \rvert > 1$. Then there exist $g \in \scal_\rbb(\rbb)$ and $\varphi \in \dcal(\rbb)$ such that  $\hscalar{\varphi}{T^{00}(g^2)\varphi} < 0$.
\end{proposition} 

\begin{proof}
We will show that, with suitable choice of $\varphi$, one has
\begin{equation}\label{eq:Xneg}
  0 >   \int d\theta \, d\eta \, \overline{\varphi(\theta)} \varphi(\eta)
 \cosh^2 \frac{\theta+\eta}{2} \, F_P(\theta-\eta)  =: X.
\end{equation}
($X$ is the expectation value of $T^{00}(0)$ up to a factor.) Rewriting \eqref{eq:t00f-expect} as
\begin{equation}\label{eq:t00g2}
  \hscalar{\varphi}{T^{00}(g^2)\,\varphi} = \frac{\mu^2}{2\pi} \int dt\, g^2(t) \int d\theta \, d\eta \, \overline{\varphi(\theta)} \varphi(\eta)
 \cosh^2 \frac{\theta+\eta}{2} \, F_P(\theta-\eta)\,e^{it\mu(\cosh\theta-\cosh\eta)},
\end{equation}
and noticing that the inner integral expression is real, continuous in $t$, and gives $X<0$ at $t=0$, it is then clear that we can choose $g$ so that \eqref{eq:t00g2} becomes negative.

To achieve \eqref{eq:Xneg}, we will choose the wave function $\varphi $ as
\begin{equation}
\varphi(\theta)=\sum_{j=1}^{2} \beta_j  \chi_{1,\rho}(\theta - \gamma_j),
\end{equation}
where $\beta_j\in\cbb$, $\gamma_j \in \rbb$ and $\rho>0$ will be specified later; the quantity $X$ then depends on these parameters. To show that $X<0$ for some $\rho>0$, it suffices to show that $X$ converges to a negative limit as $\rho \searrow 0$. Noting that $\chi_{1,\rho} (\theta) \to \delta(\theta)$ in this limit, one obtains from \eqref{eq:Xneg} that
\begin{equation}
 \lim_{\rho \searrow 0} X =  \betav^\ast  M \betav, 
\quad \text{where} \quad
 M_{jk} = \cosh^2 \frac{\gamma_j + \gamma_k}{2} F_P(\gamma_j-\gamma_k).
\end{equation} 
This expression is negative for suitable $\betav=(\beta_1,\beta_2)$ if the determinant of the matrix $M$ is negative. Setting $\gamma_1 := \gamma+\theta_P/2$, $\gamma_2 := \gamma-\theta_P/2$, with $\gamma\in\rbb$ still to be chosen, one computes
\begin{equation}
  \det M = \cosh^2 (\gamma+\theta_P/2)\cosh^2 (\gamma-\theta_P/2) - \cosh^4 \gamma F_P^2(\theta_P). 
\end{equation}
Since $ \cosh (\gamma + \theta_P/2)\cosh (\gamma - \theta_P/2) / \cosh^2 \gamma \to 1$ as $\gamma \to \infty$, and since $F_P^2(\theta_P)>1$ by assumption, $\det M$ does indeed become negative for sufficiently large $\gamma$, which concludes the proof.
\end{proof}

In particular, the conditions of Proposition~\ref{prop:negative} are met in the sinh-Gordon and the Ising models for any choice of $P$, as well as in the free model if $P \not\equiv 1$. Thus single-particle states with negative energy density exist in generic situations. 

Under stricter assumptions on the function $F_P$, we can in fact show a significantly stronger result: If $F_P$ grows stronger then a certain rate, then the negative expectation values of $T^{00}(g^2)$ become so large that quantum energy inequalities cannot hold.

\begin{proposition}\label{prop:qeinogo}
 Suppose there exist $\theta_0\geq0$ and $c>\frac{1}{2}$ such that 
 \begin{equation}\label{eq:fmingrowfast}
    \forall \theta \geq \theta_0: \quad F_P(\theta) \geq c\, \cosh \theta.
 \end{equation}
 Let $g \in \scal_\rbb(\rbb)$, $g \not\equiv 0$. Then, there exists a sequence $(\varphi_{j})_{j \in\mathbb{N}}$ in $\dcal(\rbb)$, $\lVert \varphi_{j} \rVert=1$,
 such that
 \begin{equation}\label{eq:noqei}
  \hscalar{\varphi_{j}}{T^{00}(g^2) \varphi_{j} } \to - \infty \quad \text{as } j \to \infty.
 \end{equation} 
\end{proposition}

\begin{proof}
 We set
 \begin{equation}
  \varphi_{j}(\theta) := \beta_{j,1} \, \chi_{2,\rho_j}(\theta-j) + \beta_{j,2} \,  \chi_{2,\rho_j} (\theta+j),
 \end{equation}
 where $\betav_{j} = (\beta_{j,1},\beta_{j,2}) \in \cbb^{2}$ fulfil $\lVert \betav_{j} \rVert=1$ (but are otherwise arbitrary), 
 and where $\rho_j\in(0,1)$ is a null sequence to be specified later.
 With this choice, we have $\|\varphi_j\|=1$, and one computes from \eqref{eq:t00f-expect} that
 \begin{equation}
  \hscalar{\varphi_{j}}{T^{00}(g^2) \varphi_{j} } =  \frac{\mu^2}{2 \pi} \,\betav_j^\ast M_j \betav_j \, ,
 \end{equation} 
where $M_j$ is the $2 \times 2$ matrix
\begin{equation}\label{eq:ajmn}
  M_{j,mn} =  \int d\theta\,d\eta\,
	     h_{j,mn} (\theta,\eta) \widetilde{g^2}\big(\mu k_j(\theta,\eta)\big) \,
                     \chi_{2,\rho_j}(\theta) \, \chi_{2,\rho_j}(\eta)
\end{equation}
with the functions
\begin{align}
  h_{j,11}(\theta,\eta) = h_{j,22}(\theta,\eta) &= \cosh^2\big( j + \tfrac{\theta+\eta}{2} \big) F_P(\theta-\eta),
  \\
  h_{j,12}(\theta,\eta) = h_{j,21}(\theta,\eta) &= \cosh^2\big( \tfrac{\theta-\eta}{2} \big) F_P(2j+\theta+\eta),
  \\
  k_j(\theta,\eta) &= 2 \sinh(j+\tfrac{\theta+\eta}{2}) \sinh(\tfrac{\theta-\eta}{2}).
 \end{align}
It enters here that $\chi$ is even. One of the eigenvalues of $M_j$ is $\lambda_j = M_{j,11} - M_{j,12}$, and we will show that $\lambda_j\to -\infty$, proving that \eqref{eq:noqei} holds for a suitable choice of $\betav_j$.
 
 To that end, we establish estimates on $h_{j,11}$, $h_{j,12}$ and $k_{j}$ for  $\theta,\eta\in [-\rho_j,\rho_j]$, that is,  in the region where the integrand of \eqref{eq:ajmn} is nonvanishing.  First, continuity of $F_P$ and $F_P(0)=1$ imply that $F_P(\theta-\eta) < c + \frac{1}{2}$ if $\rho_j$ is small; note that $c>\frac{1}{2}$ enters here. Estimating $\cosh^2 x \leq 1 + e^{2x}/4$ ($x \geq 0$), we then have in the relevant range for $\theta,\eta$,
 \begin{equation}\label{eq:hj11est}
  h_{j,11} (\theta,\eta)  \leq \big(1 + \tfrac{1}{4} e^{2j} e^{2 \rho_j}\big)\big(c+\tfrac{1}{2}\big).
 \end{equation}
 Further, the growth condition \eqref{eq:fmingrowfast} implies for $j > \theta_0/2+1$,
\begin{equation}\label{eq:hj12est}
   h_{j,12} (\theta,\eta)  \geq c \cosh(2j + \theta+\eta) \geq \frac{c}{2} e^{2j} e^{-2\rho_j}.
\end{equation}
Now \eqref{eq:hj11est} and \eqref{eq:hj12est} combine to give
\begin{equation}\label{eq:hjdiffest}
 h_{j,11} (\theta,\eta) - h_{j,12} (\theta,\eta) \leq c +  \frac{1}{2} +\tfrac{1}{4} e^{2j} \underbrace{\big( e^{2 \rho_j}(c+\tfrac{1}{2}) - 2 c e^{-2 \rho_j} \big)}_{\to \,\frac{1}{2} - c\, < 0}
  \leq 2c - c' e^{2j} < 0
\end{equation}
with some constant $c'>0$ and for large $j$.
Finally, since $|\theta-\eta|< 2 \rho_j<2$ in the integrand,
\begin{equation}
  \lvert k_j(\theta,\eta)\rvert \leq 2(e^j e^{(\theta + \eta)/2} )\Big( 2 \frac{|\theta-\eta|}{2}\Big) \leq 12 e^j \rho_j.
\end{equation}
 Now setting specifically $\rho_j = a e^{-j}$, with $a>0$ still to be specified, we have
 $ \lvert k_j(\theta,\eta)\rvert \leq 12a$ independent of $j$. 
 Noting that $\widetilde{g^2}(0) = \int dt\, g^2(t) > 0$, we can achieve with a suitable choice of $a$ that
 \begin{equation}\label{eq:g2lower}
   \widetilde{g^2} (\mu k_j(\theta,\eta)) \geq \frac{1}{2} \widetilde{g^2}(0) > 0.
 \end{equation}
Using \eqref{eq:hjdiffest} and \eqref{eq:g2lower} in the integrand of \eqref{eq:ajmn}, we then obtain
\begin{equation}
  \lambda_j = M_{j,11}-M_{j,12} \leq (2c - c' e^{2j}) \frac{1}{2}\widetilde{g^2}(0) \big( \lVert \chi_{2,\rho_j} \rVert_1 \big)^2
  =  \frac{a}{2} \widetilde{g^2}(0) (2c e^{-j} - c' e^{+j}) \big(\rho_j^{-1/2} \lVert \chi_{2,\rho_j} \rVert_1 \big)^2  .
\end{equation}
Here $\rho_j^{-1/2} \lVert \chi_{2,\rho_j}\rVert_1$ is actually independent of $\rho_j$. Hence $\lambda_j \to -\infty$ as $j \to\infty$, which concludes the proof.
\end{proof}

\section{Quantum energy inequalities}\label{sec:qei}

We now turn to the existence of quantum energy inequalities, i.e., we want to show that the operator $T^{00}(g^2)$ is bounded below at one-particle level. As we have seen in Proposition~\ref{prop:qeinogo}, this can be true only if the function $F_P$ does not grow too fast. The main goal of the section is the following theorem, which establishes a QEI under certain bounds on $F_P$.

\begin{theorem}\label{thm:qei}
 Suppose that there exist constants $\theta_0 \geq 0$, $\lambda_0 > 0$, and $0 < c < \frac{1}{2}$ such that
 \begin{equation}\label{eq:fpbounds}
   |F_P(\zeta) | \leq c \, \cosh \re\zeta \quad \text{whenever} \; \lvert  \re\zeta \rvert \geq \theta_0, \; \lvert \im\zeta \rvert < \lambda_0.
 \end{equation}
 Further, let $g \in\scal_\rbb(\rbb)$. Then, there exists $c_g>0$ such that
 \begin{equation}\label{eq:qitoshow}
    \forall \varphi \in \dcal(\rbb): \quad \langle  \varphi, T^{00}(g^2) \varphi \rangle \geq - c_g \|\varphi\|^2.
 \end{equation}
 The constant $c_g$ depends on $g$ (and on $F_P$, hence on $P$ and $S$) but not on $\varphi$.
\end{theorem}

The idea of the proof is as follows. In \eqref{eq:t00f-expect}, we split the integration region in both $\theta$ and $\eta$ into the positive and negative half-axis. Setting $\varphiv(\theta):=(\varphi(\theta),\varphi(-\theta))^T$ for $\theta>0$, we can rewrite our expectation value as
\begin{equation}
 X_\varphi:=\hscalar{  \varphi }{ T^{00}(g^2) \varphi } = \frac{\mu^2}{2\pi} 
 \int_0^\infty d\theta \int_0^\infty d\eta\, 
 \widetilde{g^2}(\mu\cosh\theta-\mu\cosh\eta) 
 \varphiv(\theta)^\ast M(\theta,\eta) \varphiv(\eta),
\end{equation}
where the matrix $M(\theta,\eta)$ is given by
\begin{equation}
 M(\theta,\eta) = \begin{pmatrix}
                  \cosh^2\frac{\theta + \eta}{2} F_P(\theta -\eta) & 
                   \cosh^2\frac{\theta - \eta}{2} F_P(\theta +\eta) \\
                   \cosh^2\frac{\theta - \eta}{2} F_P(\theta +\eta) &
                  \cosh^2\frac{\theta + \eta}{2} F_P(\theta -\eta)                    
                  \end{pmatrix}.
\end{equation}
The eigenvectors of $M$ are $v_+=(1,1)^T$ and $v_-=(1,-1)^T$, independent of $\theta,\eta$, and the corresponding eigenvalues are
\begin{equation}
 h_\pm(\theta,\eta) = \cosh^2\frac{\theta + \eta}{2} F_P(\theta -\eta) \pm
                   \cosh^2\frac{\theta - \eta}{2} F_P(\theta +\eta).
\end{equation}
Denoting by $\varphi_\pm$ the components of $\varphiv$ in the direction of $v_\pm$, we thus have
\begin{equation}\label{eq:iexpr}
 X_\varphi = \frac{\mu^2}{2\pi} 
 \int_0^\infty d\theta \int_0^\infty d\eta\, 
 \widetilde{g^2}(\mu\cosh\theta - \mu\cosh\eta) 
 \sum_\pm
 \overline{\varphi_\pm(\theta)} h_\pm(\theta,\eta) \varphi_\pm(\eta) .
\end{equation}
We will compare $X_\varphi$ to the following related integral expression:
\begin{equation}\label{eq:jexpr}
 Y_\varphi := \frac{\mu^2}{2\pi} 
 \int_0^\infty d\theta \int_0^\infty d\eta\, 
 \widetilde{g^2}(\mu\cosh\theta-\mu\cosh\eta) 
 \sum_\pm \overline{ \varphi_\pm(\theta)}  k_\pm(\theta)k_\pm(\eta)\varphi_\pm(\eta),
\end{equation}
where
\begin{equation}\label{eq:kpmdef}
 k_\pm(\theta) := \sqrt{| h_{\pm}(\theta,\theta)  |}
= \sqrt{| \cosh^2\theta \pm F_P(2\theta)  |}.
\end{equation}
Specifically, we will show that $Y_\varphi \geq 0$ and that $|X_\varphi-Y_\varphi|$ is bounded in $\|\varphi\|^2$. We do this in several steps; the hypothesis of the theorem is always assumed.

\begin{lemma}\label{lem:intpos}
For any $\varphi$, we have $Y_\varphi \geq 0$.
\end{lemma}

\begin{proof}
Using the identity 
\begin{equation}\label{tildefidentity}
\widetilde{g^2}(p-p')
 =\int \frac{dq}{2\pi}\,\tilde{g}(q +p)\overline{\tilde{g}(q + p')},
\end{equation}
we can rewrite the integral $Y_\varphi$ as
\begin{equation}
 Y_\varphi = \frac{\mu^2}{4\pi^2} 
 \sum_\pm \int dq \,
 \Big\lvert \int_0^\infty d\eta \, 
  a_\pm(\eta,q) \Big\rvert^2
 \quad\text{with} \quad
 a_\pm(\eta,q) := k_\pm(\eta)\varphi_\pm(\eta) \overline{\tilde{g}(q+\mu\cosh\eta)} .
\end{equation}
But this is clearly nonnegative.
\end{proof}

For estimating $|X_\varphi-Y_\varphi|$, we first need an estimate for the relevant integral kernels,
into which the growth bound \eqref{eq:fpbounds} for $F_P$ will crucially enter.

\begin{lemma}\label{lem:kernelest}
  Set
  \begin{equation}
 \ell_\pm(\rho,\tau) := h_\pm(\rho+\tfrac{\tau}{2}, \rho-\tfrac{\tau}{2}) - k_\pm(\rho+\tfrac{\tau}{2}) k_\pm(\rho-\tfrac{\tau}{2}).
\end{equation}
Then, there exists $a>0$ such that for all $\rho \geq \theta_0+1$ and $\tau \in [-1,1]$,
\begin{equation}
 |\ell_\pm(\rho,\tau)| \leq a \tau^2 \cosh^2 \rho.
\end{equation}
\end{lemma}

\begin{proof}
One notes that $\ell_\pm(\rho,0)=0$ and that $\ell_\pm(\rho,\tau)$ is symmetric in $\tau$. A Taylor expansion of $\ell_\pm(\rho,\tau)$ in $\tau$ around $\tau = 0$ then yields that
\begin{equation}\label{eq:taylorest}
 |\ell_\pm(\rho,\tau)| \leq \frac{\tau^2}{2} \sup_{|\xi|\leq 1} \Big|\frac{\partial^2 \ell_\pm}{\partial \tau^2 }(\rho,\xi)\Big|.
\end{equation}
Thus our task is to estimate the derivative. 
As a first step, we remark that Cauchy's formula allows us to deduce estimates for the derivatives of $F_P$ from \eqref{eq:fpbounds}: One finds constants $c',c''>0$ such that
\begin{equation}\label{eq:fpderivbounds}
  \forall \theta \geq \theta_0+1: \quad
   \Big|\frac{dF_P}{d\theta}(\theta) \Big| \leq c' \cosh \theta, \;\;
   \Big|\frac{d^2 F_P}{d\theta^2}(\theta) \Big| \leq c'' \cosh \theta.
\end{equation}
Now we explicitly compute
\begin{equation}
\begin{aligned}
  \frac{\partial^2}{\partial \tau^2} h_\pm(\rho+\tfrac{\tau}{2}, \rho-\tfrac{\tau}{2})
  &= \frac{\partial^2}{\partial \tau^2} \Big(\cosh^2 \rho \, F_P(\tau) \pm \cosh^2 \tfrac{\tau}{2} F_P(2\rho) \Big)\\
  &= \cosh^2 \rho\, \frac{d^2 F_P}{d\theta^2}(\tau) \pm \frac{1}{2} \cosh\tau \, F_P(2\rho).
  \end{aligned}
\end{equation}
In the region $|\tau| \leq 1$, $\rho \geq \theta_0+1$, we therefore have due to \eqref{eq:fpbounds}, \eqref{eq:fpderivbounds},
\begin{equation}\label{eq:hppest}
   \Big\lvert \frac{\partial^2}{\partial \tau^2} h_\pm (\rho+\tfrac{\tau}{2}, \rho-\tfrac{\tau}{2}) \Big\rvert \leq a_1 \cosh^2 \rho
\end{equation}
with some $a_1>0$. 
For the derivative of the second term in $\ell_\pm$, we obtain
\begin{equation}\label{eq:Kderiv}
 \frac{\partial^2}{\partial \tau^2} k_\pm(\rho+\tfrac{\tau}{2}) k_\pm(\rho-\tfrac{\tau}{2})
 = \frac{1}{4} \frac{d^2 k_\pm}{d\theta}(\rho+\tfrac{\tau}{2}) k_\pm(\rho-\tfrac{\tau}{2}) 
 + \frac{1}{4} k_\pm(\rho+\tfrac{\tau}{2}) \frac{d^2 k_\pm}{d\theta}(\rho-\tfrac{\tau}{2}) .
\end{equation}
Note here that the radicand in $k_\pm(\theta)$ -- cf.~\eqref{eq:kpmdef} -- is actually positive for $\theta\geq \theta_0$ due to \eqref{eq:fpbounds}, thus the function is differentiable. ($c<\frac{1}{2}$ enters here.) 
For $\theta \geq \theta_0$, we then find $|k_\pm(\theta)| \leq a_2 \cosh \theta $. For its second derivative, we have
\begin{equation}
\begin{aligned}
   \frac{d^2 k_\pm}{d\theta^2}(\theta) 
    &= \frac{d^2}{d\theta^2}
    \sqrt{\cosh^2 \theta \pm F_P(2\theta)}  \\
    &= - \frac{(\frac{1}{2} \sinh(2\theta) \pm \frac{dF_P}{d\theta}(2\theta))^2}{ (\cosh^2\theta \pm F_P(2\theta))^{3/2}}
    +  \frac{\cosh(2\theta) \pm 2 \frac{d^2 F_P}{d\theta^2}(2\theta)}{ (\cosh^2\theta \pm F_P(2\theta))^{1/2}} .
    \end{aligned}
\end{equation}
In the denominator, we can estimate $\cosh^2\theta \pm F_P(2\theta) \geq a_3 \cosh^2 \theta$ (for $\theta\geq \theta_0$, with $a_3>0$); again, $c < \frac{1}{2}$ enters. In the numerator, we estimate the derivatives of $F_P$ by \eqref{eq:fpderivbounds}. This yields
\begin{equation}\label{eq:kppest}
      \Big\lvert \frac{d^2 k_\pm}{d \theta^2}  \Big\rvert  \leq a_4 \cosh \theta.
\end{equation}
Combining \eqref{eq:hppest}, \eqref{eq:Kderiv} and \eqref{eq:kppest}, we find for $\rho \geq \theta_0+1$, $|\tau| \leq 1$,
\begin{equation}\label{eq:deriv2est}
 \Big\lvert \frac{\partial^2 \ell_\pm}{\partial \tau^2}(\rho,\tau)\Big\rvert
 \leq a_5 \cosh^2\rho
\end{equation}
with some $a_5>0$. Applying this estimate in \eqref{eq:taylorest} yields the desired result.
\end{proof}

This estimate on the integral kernel allows us to find bounds for $|X_\varphi-Y_\varphi|$.

\begin{lemma}\label{lem:inths}
For every $g \in \scal_\rbb(\rbb)$, there exists $c_g>0$ such that
$|X_\varphi-Y_\varphi| \leq c_g \|\varphi\|^2$ for any $\varphi\in\dcal(\rbb)$.
\end{lemma}

\begin{proof}
Comparing \eqref{eq:iexpr} and \eqref{eq:jexpr} and using Cauchy-Schwarz, we can choose
\begin{equation}\label{eq:toconverge1}
 c_g := \sum_\pm \int_0^\infty d\theta \int_0^\infty d\eta\, 
 \big\lvert \widetilde{g^2}(\mu\cosh\theta-\mu\cosh\eta) \big\rvert^2 \,
 \big\lvert h_\pm(\theta,\eta) - k_\pm(\theta) k_\pm(\eta)  \big\rvert^2
\end{equation}
if we can show that these integrals converge. We change to new variables: $\rho = \frac{\theta+\eta}{2}$, $\tau = \theta-\eta$,  $|\partial (\rho,\tau)/\partial(\theta,\eta)|=1$. The integral now runs over $0 < \rho < \infty$, $-2\rho < \tau < 2\rho$. However, the integral over the region $0 < \rho < \theta_0+1$ certainly exists since the integrand is bounded there. Also the integral over $\rho > \theta_0+1$, $1 < |\tau| < 2 \rho$ is finite, since $h_\pm, k_\pm$ are bounded by a polynomial in $\cosh \rho$, whereas \begin{equation}\label{eq:g2}
 \big\lvert \widetilde{g^2}( \mu\cosh\theta-\mu\cosh\eta) \big\rvert 
 =  \big\lvert \widetilde{g^2}(2 \mu \sinh \rho \sinh  \tfrac{\tau}{2}) \big\rvert
\end{equation}
decays faster then any power of $\cosh\rho$ there, as $g$ is of Schwartz class.
Thus it remains to prove convergence of
\begin{equation}\label{eq:toconverge2} 
 \int_{\theta_0+1}^\infty d\rho \int_{-1}^1 d\tau\, 
 \big\lvert \widetilde{g^2}(2 \mu \sinh \rho \sinh  \tfrac{\tau}{2}) \big\rvert^2 \,
 \big\lvert \ell_\pm(\rho,\tau) \big\rvert^2 
\end{equation}
with $\ell_\pm$ as in Lemma~\ref{lem:kernelest}. We estimate \eqref{eq:g2} more carefully for $\rho > \theta_0+1$, $ |\tau| < 1$:
\begin{equation}\label{eq:g2upper}
 \big\lvert \widetilde{g^2}( \mu\cosh\theta-\mu\cosh\eta) \big\rvert 
  \leq \frac{b}{\tau^4 \cosh^4 \rho + 1}
\end{equation}
with some $b>0$. Using this and the estimates from Lemma~\ref{lem:kernelest},  
we have 
\begin{equation}
 \eqref{eq:toconverge2}  \leq a^2 b^2 \int_{\theta_0+1}^\infty d\rho \int_{-1}^1 d\tau 
   \frac{\tau^4 \cosh^4\rho }{(\tau^4 \cosh^4 \rho + 1)^2} 
   \leq a^2 b^2 \int_{\theta_0+1}^\infty \frac{d\rho}{\cosh\rho} \int_{-\infty}^\infty d\kappa 
   \frac{\kappa^4 }{(\kappa^4 + 1)^2} < \infty.
\end{equation}
(We have used the substitution $\kappa = \tau \cosh \rho$.) This means that $c_g$ in \eqref{eq:toconverge1} is finite.
\end{proof}

Our main result is now a direct consequence of the lemmas above.

\begin{proof}[Proof of Theorem~\ref{thm:qei}]
Using Lemmas~\ref{lem:intpos} and \ref{lem:inths}, we can compute
\begin{equation}
   \hscalar{\varphi}{T^{00}(g^2)\,\varphi} = Y_\varphi + (X_\varphi-Y_\varphi) \geq Y_\varphi - |X_\varphi-Y_\varphi| \geq -c_g \|\varphi\|^2. \qquad \qed
\end{equation}
\noqed
\end{proof}

Note that we have actually shown that $T^{00}(g^2)$ is a sum of a positive quadratic form and a Hilbert-Schmidt operator, since we estimated $|X_\varphi-Y_\varphi|$ by showing that the integral kernel involved is an $L^2$ function. But we will not need this detail for our present purposes.

\section{Uniqueness of the energy density in certain models}\label{sec:models}

Our results so far indicate that the existence or non-existence of quantum energy inequalities at one-particle level is determined by the asymptotic behaviour of the function $F_P(\theta)=P(\cosh\theta) \Fmin(\theta+i\pi)$. Essentially, Theorem~\ref{thm:qei} says that QEIs hold if $F_P(\theta) \lesssim \frac{1}{2} \cosh \theta$, and Proposition \ref{prop:qeinogo} shows that QEIs do \emph{not} hold if $F_P(\theta)  \gtrsim \frac{1}{2} \cosh \theta$. Let us now see what this means in specific models.

Let us first consider the Ising model, with the function $\Fmin$ given by
\begin{equation}
  \Fmin(\zeta) = -i \sinh \frac{\zeta}{2}.
\end{equation}
If now $P=1$, then $F_P(\theta) \sim \cosh^{1/2} \theta$ for large $\theta$ (also on a small complex strip), the conditions of Theorem~\ref{thm:qei} are fulfilled, and the energy density satisfies a QEI. If, on the other hand, $P$ is a polynomial of degree $d \geq 1$, then $F_P(\theta) \sim \cosh^{d+1/2} \theta$, the conditions of Proposition \ref{prop:qeinogo} are met, and hence no QEI can hold. 
In other words, in the Ising model, the one-particle energy density is \emph{uniquely determined} by Proposition \ref{prop:stressenergy} and the additional requirement that a QEI holds. We summarize this in the following theorem.

\begin{theorem}\label{thm:uniqueness-ising}
In the Ising model (cf.~Table~\ref{tab:models}), a one-particle quantum energy inequality holds if, and only if, $P=1$ in Proposition \ref{prop:stressenergy}. \hfill \qed
\end{theorem}

We remark that in the case where no QEI holds, Proposition~\ref{prop:qeinogo} actually implies that $T^{00}(g^2)$ is unbounded below for \emph{any} real-valued Schwartz function $g\not\equiv 0$ -- we cannot even establish a QEI by restricting to a smaller class of smearing functions.

In the sinh-Gordon model, the situation is more complicated. Here the function $\Fmin$ is of the form
\begin{equation}
  \Fmin(\zeta) =  \exp J_{B}(\zeta).
\end{equation}
Since $J_B(\theta+i\lambda)$ converges to a constant as $\theta\to\pm\infty$, uniformly on a strip around $\lambda=\pi$, we know that $|\Fmin(\theta+i\lambda)|$ is bounded there and does not converge to 0 at infinity. 
If now $P=1$, then $F_P$ is bounded on the strip, thus Theorem~\ref{thm:qei} applies and yields a QEI. If $\deg P \geq 2$, then $\lvert F_P (\theta) \rvert \sim (\cosh \theta)^{\deg P}$ grows faster than $c \cosh\theta$ for any $c$, and by Proposition \ref{prop:qeinogo} no QEI can hold. If however  $\deg P=1$, then details of the polynomial matter. Since $P(1)=1$, such $P$ must be of the form
\begin{equation}
   P(x) = (1-\nu) + \nu x \quad \text{with some } \nu \in \rbb.
\end{equation}
The existence of a QEI now depends on $\nu$. Setting $\Fmin^\infty := \lim_{\theta \to\pm\infty} \Fmin(\theta+i\pi)$, we have
\begin{equation}
  F_P(\theta) =  (1-\nu)\Fmin(\theta+i\pi)  + \nu \Fmin^\infty \cosh\theta   \frac{\Fmin(\theta+i\pi)}{\Fmin^\infty}.
\end{equation}
The first summand is negligible against $\cosh\theta$ for large $\theta$, and the last fraction becomes arbitrarily close to 1. Therefore, if $|\nu| > 1/2\Fmin^\infty$, then Proposition \ref{prop:qeinogo} applies with $c = \frac{1}{2}|\nu|\Fmin^\infty + \frac{1}{4} > \frac{1}{2}$. No QEI can hold. If, on the other hand,   $|\nu| < 1/2\Fmin^\infty$, then similarly Theorem \ref{thm:qei} applies with $c = \frac{1}{2}|\nu|\Fmin^\infty + \frac{1}{4} < \frac{1}{2}$, and a QEI follows. In the borderline case $|\nu| = 1/2\Fmin^\infty$, our methods do not yield a result. 

Interestingly, the behaviour of the free theory, where $\Fmin=1$, exactly matches the one of the sinh-Gordon model, only with $\Fmin^\infty=1$. We summarize:
   
\begin{theorem}\label{thm:uniqueness-sinh}
In the free and the sinh-Gordon models (cf.~Table~\ref{tab:models}),  a one-particle quantum energy inequality holds if
\begin{equation}\label{eq:palpha}
   P(x) = (1-\nu) + \nu x \quad \text{with } \nu \in \rbb, \; |\nu| < \frac{1}{2 \, \Fmin^\infty},
\end{equation}
where $\Fmin^\infty = \lim_{\theta \to\pm\infty} \Fmin(\theta+i\pi)$.
If $P$ is of this form but with  $|\nu| > 1/2\Fmin^\infty$, or if $\deg P \geq 2$, then no quantum energy inequality holds. \hfill\qed
\end{theorem}

Thus, in the sinh-Gordon model and the free field, demanding a QEI strongly restricts the choice for the energy density, but it does not fix the energy density uniquely. Even at one-particle level, a certain arbitrariness remains.
At least in the free field situation, we observe however that the ``canonical'' case $P=1$ is distinguished by the fact that it is the only choice which makes the energy density \emph{nonnegative} in one-particle states.

The generalized Ising and generalized sinh-Gordon models in Table~\ref{tab:models} show a similar behaviour, but allow for some more variety, since the factors $-i\sinh(\zeta/2)$ in the denominator of $\Fmin$ make the function decay at real infinity. In consequence, for $\deg P$ not too large, QEIs will hold due to Theorem~\ref{thm:qei}; in fact the hypothesis of Proposition~\ref{prop:negative} may be violated and the energy density may be positive at one-particle level. For sufficiently large $\deg P$, however, Proposition~\ref{prop:qeinogo} will apply and hence QEIs will break down.

\section{Numerical results}\label{sec:numerical}

In the previous sections, we have discussed whether expectation values of the energy density can be negative, and we established under which conditions there are lower bounds to the energy density, i.e., $T^{00}(g^2) \geq -c_g \idop$ at one-particle level. Clearly it would be of interest to compute what the \emph{best} lower bound $c_g$ is for given $g$, or in other words, what the lowest spectral value of the integral operator in question is. However, all our methods so far are based on estimates, which are in fact quite inexplicit at places, and are therefore not suited to compute a useful value for the best lower bound. 
The only feasible way to do so appears to be numerical approximation. In this section, we will describe numerical methods to that end.

To be precise, we consider the operator $\opt := P_1T^{00}(g^2)P_1$, where $P_1$ is the projection onto	 the one-particle space. Our task is to compute the lowest value in the spectrum of $\opt$. (So far, we had defined $\opt$ as a symmetric operator on $\dcal(\rbb)$ only. However, since by Theorem~\ref{thm:qei} it is bounded below -- we assume that the hypothesis of the theorem is fulfilled --, $\opt$ canonically extends to a selfadjoint operator with the same lower bound \cite[Thm.~X.23]{ReedSimon:1975-2}, denoted by the same symbol, so that $\opt$ has a well-defined spectral decomposition.)

As a first complication, while we have shown that $\opt$ is bounded below under certain conditions, it is certainly unbounded from above due to the strong growth of the integral kernel at large arguments, cf. Eq.~\eqref{eq:t00f-expect}. We need to remove this discontinuity before the operator becomes amenable to a stable numeric treatment. To that end, we introduce a ``rapidity cutoff'': we will restrict the support of the wave functions to an interval $[-R,R]$, with sufficiently large $R>0$; in other words, we consider  $\opt$ as an integral operator on $L^2([-R,R],d\theta)$ rather than on $L^2(\rbb,d\theta )$. On this restricted space, the operator is bounded, in fact it is compact. The cutoff will certainly affect the large positive spectral values of the operator; but we claim that it does not significantly change the lowest spectral value, which we are interested in here. That this claim is at least self-consistent will become clearer below.

As a next step, we discretize the integral operator. To that end, we use an orthogonal Galerkin method with piecewise constant bump functions and application of the midpoint rule; alternatively speaking, the Nystr\"om method associated with the midpoint rule. For details of these approximation algorithms see, e.g., \cite[Ch.~4]{Chatelin:spectral}; but for the benefit of the reader, let us give a rough description of the method here. We divide the interval $[-R,R]$ into $N$ equally sized subintervals of length $h=2R/N$, with midpoints $\theta_j = -R + (j+\frac{1}{2})h$, where $j = 0,\ldots,N-1$. We consider the orthonormal system of step functions supported on these intervals,
\begin{equation}
  \varphi_j (\theta) = \begin{cases}
                           h^{-1/2} \quad & \text{if } |\theta_j-\theta| < \frac{h}{2} \\  
	                 0 & \text{otherwise}
                      \end{cases},
\qquad 0 \leq j < N.
\end{equation}
Then, we will only consider the matrix elements
\begin{equation}\label{eq:mtxe-numeric}
     M_{jk} = \hscalar{\varphi_j}{ T^{00}(g^2)  \varphi_k } = \frac{\mu^2}{2 \pi h}
 \int\limits_{\theta_j-\frac{h}{2}}^{\theta_j+\frac{h}{2}} d\theta
 \int\limits_{\theta_k-\frac{h}{2}}^{\theta_k+\frac{h}{2}} d\eta\,
 \cosh^2 \frac{\theta+\eta}{2} \,  F_P(\theta-\eta) \widetilde{g^2}(\mu\cosh\theta-\mu\cosh\eta).
\end{equation}
For easier evaluation, we will approximate the integral using the midpoint rule, which yields
\begin{equation}
    M_{jk} \approx \hat  M_{jk} :=  \frac{\mu^2}{2 \pi} h 
 \cosh^2 \frac{\theta_j+\theta_k}{2} \,  F_P(\theta_j-\theta_k) \widetilde{g^2}(\mu\cosh\theta_j-\mu\cosh\theta_k).
\end{equation}
This expression is then straightforward to compute, although evaluating the function $F_P$ will in some cases need numerical integration techniques for approximating $J_B$, cf.~\eqref{eq:jdef}. The eigenvalues and eigenvectors of the real symmetric $N\times N$ matrix $\hat M$ can now be found with standard numerical methods, such as the implicit QL algorithm.  While we have not given full error estimates here, we can reasonably expect that the lower eigenvectors and eigenvalues of $\hat M$ give an  approximation of the corresponding properties of $\opt$. The computer code implementing the method, which was used to produce the results in the following, is provided with this article \cite{suppl:code}, along with documentation.

For the sake of concreteness, we will in the following fix the smearing function $g$ to be a Gaussian,
\begin{equation}\label{eq:gexample}
 g(t) = \pi^{-1/4} \sqrt{\frac{\mu}{2\sigma}}  \exp\Big( -\frac{\mu^2 t^2}{8 \sigma^2}\Big)
\end{equation}
with a dimensionless parameter $\sigma>0$, so that $  \int g(t)^2 \,dt = 1 $
and
\begin{equation}
 \widetilde{g^2}(p) = \exp \Big(-\frac{\sigma^2 p^2}{\mu^2} \Big).
\end{equation}
With this, it is clear that the matrix elements \eqref{eq:mtxe-numeric} scale quadratically with $\mu$, so that we only
need to consider the ``normalized'' case $\mu = 1$; in other words, all energy density values will be plotted in units of $\mu^2$ in the following.

\begin{figure}
     \centering
     \begin{subfigure}[t]{0.5\textwidth}
             \includegraphics[width=\textwidth]{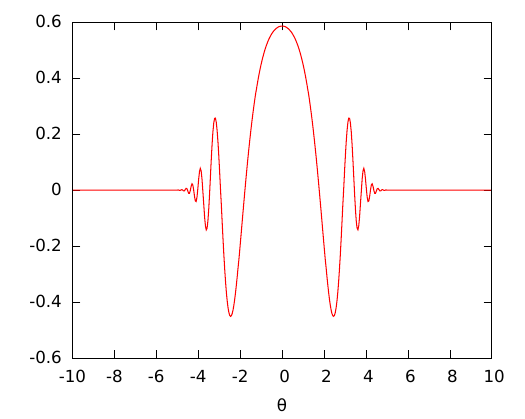}
             \caption{Ising model}
             \label{fig:eigenvec-ising}
     \end{subfigure}%
     \begin{subfigure}[t]{0.5\textwidth}
              \includegraphics[width=\textwidth]{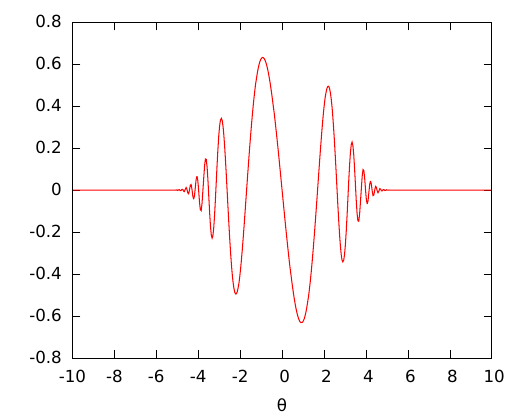}
             \caption{Sinh-Gordon model with coupling $B=1$.}
             \label{fig:eigenvec-sinh}
     \end{subfigure}
     \caption{Lowest eigenvector of $\opt$. Parameters: $N=500$, $R=10$, $\sigma = 0.1$, $P=1$.}
     \label{fig:eigenvec}
\end{figure}

We will first discuss the form of the eigenvector for the lowest eigenvalue of $\hat M$, which turns out to be clearly separated from the next higher eigenvalue. Figure \ref{fig:eigenvec} shows the eigenvectors (eigenfunctions) for the Ising model and for the sinh-Gordon model with maximal coupling (i.e., $B=1$), both with the ``canonical'' form of the energy density ($P=1$). While the vector certainly differs between these models, we can observe in both cases that the eigenfunction decays rapidly for large $\theta$, much before the chosen cutoff (here $R=10$) is reached. Hence our assumption that the lowest eigenvector is not influenced by the (reasonably large) cutoff is self-consistent. 

This observation also yields a useful consistency check for the numerical method: By testing whether the components of the eigenvector in direction $\varphi_j$  are small for $j \approx 0$ and $j \approx N-1$,  we can verify whether our choice for the cutoff $R$ was appropriate; see the method  \\ {\tt kernels.OneParticleKernel.checkPlausibility()} in the computer code. 
We remark that a suitable choice of $R$ is crucial to the stability of the numeric method: The highest (positive) eigenvalues of  $\hat M$ are several orders of magnitude larger than the lowest (negative) one, and they grow rapidly with increasing $R$. Hence if $R$ is chosen unduly large, significant roundoff errors will affect the result of the QL algorithm when working at a given floating point precision. 

\begin{figure}[btp]
     \centering
     \begin{subfigure}[b]{0.5\textwidth}
             \includegraphics[width=\textwidth]{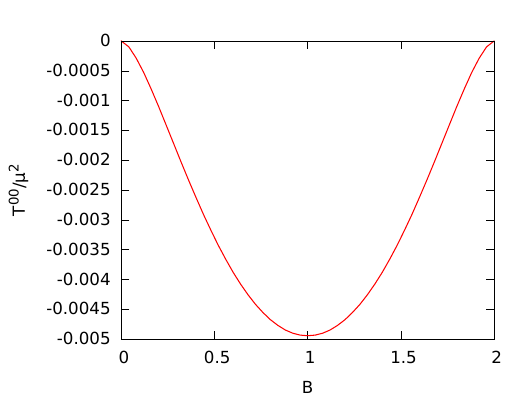}
             \caption{Sinh-Gordon model, coupling $0 < B < 2$}
             \label{fig:coupling-sinhgordon}
     \end{subfigure}%
     \begin{subfigure}[b]{0.5\textwidth}
             \includegraphics[width=\textwidth]{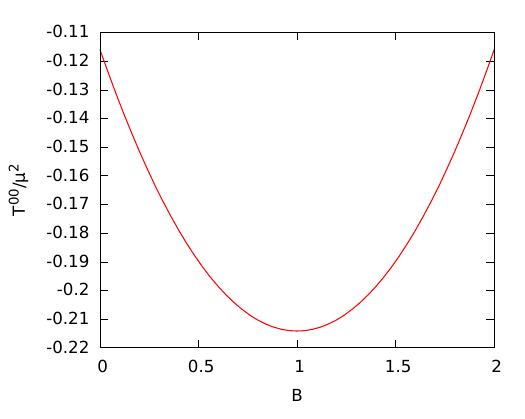}
             \caption{Generalised Ising model, coupling $0 < B < 2$}
             \label{fig:coupling-sinhising}
     \end{subfigure}
     \caption{Lowest eigenvalue of $\opt$ depending on the coupling constant. Parameters: $N=500$, $R=7$, $\sigma = 0.1$; (a) $P=1$, (b) $P(x)=(1+x)/2$.}
     \label{fig:coupling}
\end{figure}

Next, let us compare the lowest eigenvalue of $\hat M$ in various models. Specifically, let us choose the sinh-Gordon model, depending on the coupling constant $0<B<2$, 
again with $P=1$. Figure \ref{fig:coupling-sinhgordon} shows the lowest eigenvalue in this model as a function of the coupling constant $B$. Note that $B$ and $2-B$ correspond to the same scattering function. As expected, when $B$ approaches the values 0 or 2, the lowest eigenvalue approaches the one of the free theory (namely, 0). The lowest possible negative energy density is reached when the coupling is maximal ($B=1$), which fits with the picture that negative energy density in one-particle states is an effect of self-interaction in the quantum field theory.

We also investigate the corresponding generalised Ising model with one coupling constant, i.e., with scattering function
\begin{equation}
S(\zeta) = -\dfrac{\sinh \zeta - i \sin B \pi/2 }{\sinh \zeta + i \sin B \pi/2 } = - S_\text{sinh-Gordon}(\zeta), \quad 0<B<2.
\end{equation}
Note that $S(\zeta) \to S_\mathrm{Ising}(\zeta)=-1$ as $B \to 0$.
In this case, rather than $P=1$ we choose $P(x)=(1+x)/2$: this guarantees that also $F_{P}$ converges to the canonical Ising model expression as $B \to 0$.
Both Proposition~\ref{prop:negative} and Theorem~\ref{thm:qei} are applicable with this choice of $P$. As Fig.~\ref{fig:coupling-sinhising} shows, the lowest possible eigenvalue is still reached at maximal coupling $B=1$, whereas for $B\to 0$ we obtain the Ising value of $\approx-0.116$. 

\begin{figure}[htpb]
     \centering
     \begin{subfigure}[t]{0.5\textwidth}
             \includegraphics[width=\textwidth]{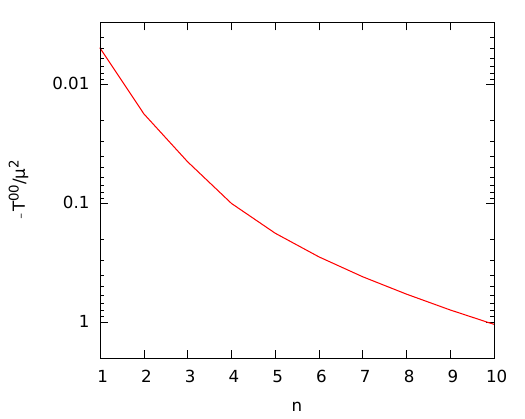}
             \caption{Several real coupling constants with $B=1$}
             \label{fig:multiplecoupling-real}
     \end{subfigure}%
     \begin{subfigure}[t]{0.5\textwidth}
             \includegraphics[width=\textwidth]{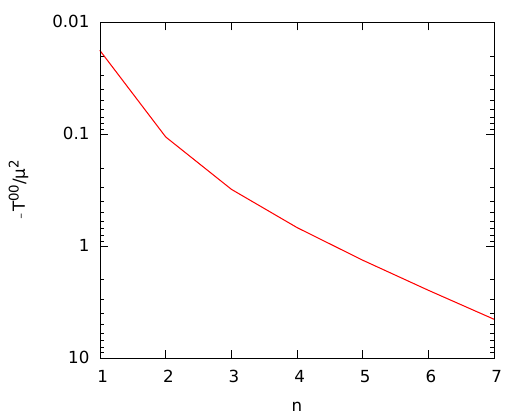}
             \caption{Several complex coupling constants $B_k=1+k \cdot 0.1 i$, $1 \leq k \leq n$, with complex conjugate factors added. }
             \label{fig:multiplecoupling-complex}
     \end{subfigure}
     \caption{Lowest eigenvalue of $\opt$ in the generalized sinh-Gordon model. The horizontal axis shows the number of coupling constants $n$. Parameters: $N=500$, $R=7$, $\sigma=0.1$; see Eq.~\eqref{eq:pmultiple} for $P$. Note the logarithmic scale on the vertical axis.}
     \label{fig:multiplecoupling}
\end{figure}

To investigate other types of interaction, we also show the lowest eigenvalue in the generalized sinh-Gordon model with several factors in the scattering matrix, i.e., with several coupling constants $B_j$ (cf. Table \ref{tab:models}). We do this once for $n$ coupling constants that are equal and real ($B_j = 1$ for all $j$), which is shown in Fig.~\ref{fig:multiplecoupling-real},
and once for pairs of complex-conjugate factors with $B=1 \pm k \cdot 0.1 i$, $k=1,\ldots,n$ see Fig.~\ref{fig:multiplecoupling-complex}. 
We choose the polynomial $P$ such that the free-field expression $F_P=1$ is reached when all coupling constants are set to 0; namely
\begin{equation}\label{eq:pmultiple}
\begin{aligned}
  P(x) &= \Big(\frac{1+x}{2}\Big)^{\lfloor n/2 \rfloor} \quad &\text{($n$ real coupling constants)} \quad &\text{and} \quad
\\
P(x) &= \Big(\frac{1+x}{2}\Big)^{n} \quad &\text{($n$ complex conjugate pairs)}.&
\end{aligned}
\end{equation}
In both cases, when the number $n$ grows, the lowest eigenvalue of $\opt$ rapidly grows towards negative values. This is consistent with our computations in Sec.~\ref{sec:qei}, since in particular the derivatives of $F_P$ grow rapidly with $n$. Thus again, a stronger interaction leads to more negative energy density.

\begin{figure}[htpb]
\begin{center}
  \includegraphics[width=0.8\textwidth]{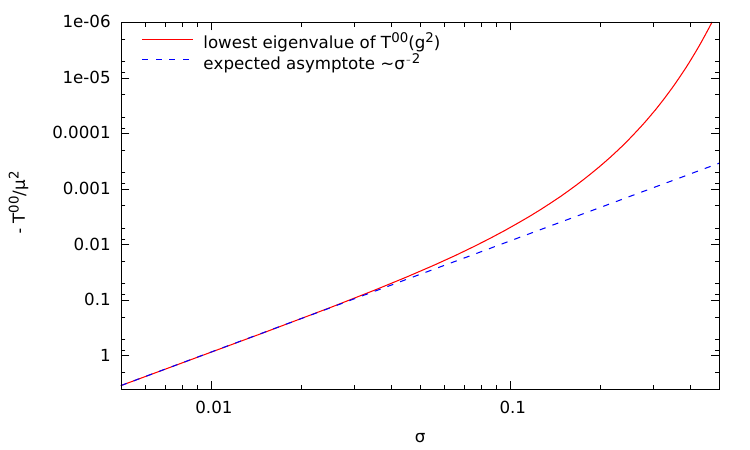}
 \caption{Lowest eigenvalue in the sinh-Gordon model ($B=1$) depending on the smearing width parameter $\sigma$. Parameters: $N=500$, $R=10$, $P=1$. Note the double logarithmic scale}
 \label{fig:smearingwidth}
\end{center}
\end{figure}

So far, we have always used the same ``smearing function'' $g^2$ in our analysis. We now vary the width $\sigma$ of the Gaussian function $g$, see Eq.~\eqref{eq:gexample}, and observe the change in the lowest eigenvalue of the smeared energy density, as shown in Fig.~\ref{fig:smearingwidth}. Since the energy density is expected to be a quantum field of scaling dimension 2 in the high energy limit, one conjectures that for small values of $\sigma$, the lowest eigenvalue of $\opt$ roughly scales with $\sigma^{-2}$. In fact, this is confirmed by our numerical results; a corresponding linear asymptote with slope $2$ has been added to the double-logarithmic plot Fig.~\ref{fig:smearingwidth} to demonstrate this.

\begin{figure}[htpb]
     \centering
     \begin{subfigure}[t]{0.5\textwidth}
             \includegraphics[width=\textwidth]{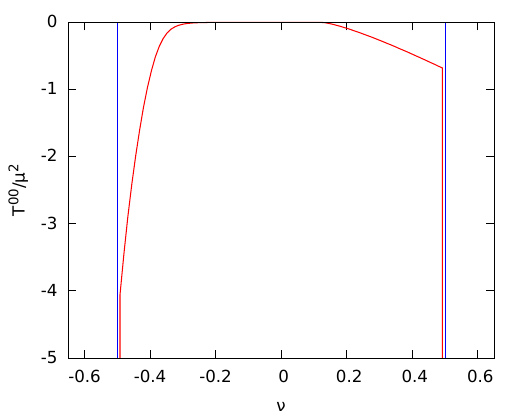}
             \caption{Free field}
             \label{fig:nonstandard-free}
     \end{subfigure}%
     \begin{subfigure}[t]{0.5\textwidth}
             \includegraphics[width=\textwidth]{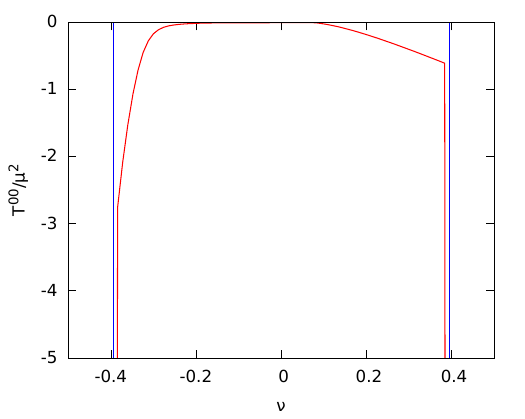}
             \caption{sinh-Gordon model, $B=1$}
             \label{fig:nonstandard-sinhgordon}
     \end{subfigure}
     \caption{Lowest eigenvalue of $\opt$ with a ``non-canonical'' energy density with $P(x)=(1-\nu)+\nu x$.  Parameters: $N=300$, $R=8$, $\sigma=0.1$. The vertical lines indicate the theoretical
      limit for existence of QEIs, $|\nu|=1/2\Fmin^\infty$. Note that in case (b), we have $\Fmin^\infty \approx 1.267$.}
     \label{fig:nonstandard}
\end{figure}

Finally, in the free model and the sinh-Gordon model, we want to investigate the ``noncanonical'' form of the energy density as described in Sec.~\ref{sec:models}, where the polynomial $P$ is not constant but linear, see Eq.~\eqref{eq:palpha}. In this case, we know from Theorem~\ref{thm:uniqueness-sinh} that $\opt$ is bounded below if $|\nu|<1/2\Fmin^\infty$, but unbounded if $|\nu|>1/2\Fmin^\infty$. Our numerical results in Fig.~\ref{fig:nonstandard} confirm this: 
At the critical value $|\nu|=1/2\Fmin^\infty$, indicated by the vertical lines, the lowest eigenvalue of $\hat M$ becomes several orders of magnitude larger than for smaller values of $|\nu|$. In fact, for $|\nu|>1/2\Fmin^\infty$, the form of the eigenvector changes drastically from what is known from Fig.~\ref{fig:eigenvec}, and the eigenvalue is highly dependent on the value of the cutoff $R$. (Details of this are not shown in the graph.)

\section{Conclusions and outlook}\label{sec:conclusions}

In this article, we have investigated the energy density in certain integrable models at one-particle level. Starting from the general properties of a local operator, as derived in \cite{BostelmannCadamuro:characterization}, and the generic properties of a stress-energy tensor, we were able to determine $T^{\alpha\beta}$ in one-particle matrix elements up to a certain polynomial factor, $P(\cosh(\theta-\eta))$. While one-particle states with negative energy density quite generally exist, the choice of energy density is further restricted if we demand that a quantum energy inequality holds. 

In fact, in the Ising model, the existence of one-particle state-independent QEIs fixes the form of the energy density in one-particle states uniquely. In a different class of models, including the sinh-Gordon model and the free Bose field, QEIs reduce the ambiguity in the energy density, but they still leave us with one parameter to be chosen (see Theorem~\ref{thm:uniqueness-sinh}). This choice affects the form factor $F_2$ of the stress-energy tensor, and is different from the ambiguity in $F_1$ discussed in \cite{MussardoSimonetti:1994}, which does not contribute to one-particle expectation values.

It is interesting to note how this choice in $F_2$ works out in the free field case ($S=1$). The usual ``canonical'' stress-energy tensor, 
\begin{equation}
T^{\alpha\beta}_{\text{canonical}} = {:} \partial^\alpha \phi \, \partial^\beta \phi {:} 
 - \frac{1}{2} \eta^{\alpha\beta}{:} \partial^\gamma \phi \,\partial_\gamma \phi {:} 
 + \frac{\mu^2}{2} \eta^{\alpha\beta} {:} \phi^2{:}\;, 
\end{equation}
corresponds to our choice $P(x)=1$ at one-particle level. But other choices of $P$ can be linked with local fields on Fock space as well. In particular, one possible choice of stress-energy tensor corresponding to $P(x)=x$ is
\begin{equation}
T^{\alpha\beta}_{\text{alternative}} = 
-{:} \phi \, \partial^\alpha \partial^\beta \phi {:} 
 + \frac{1}{2} \eta^{\alpha\beta}{:} \partial^\gamma \phi \, \partial_\gamma \phi {:} 
 - \frac{\mu^2}{2} \eta^{\alpha\beta} {:} \phi^2{:} \,.
\end{equation}
One can easily show that this field does in fact satisfy the standard properties of a stress-energy tensor on the full Fock space, including covariance and the continuity equation. A permissible energy density that fulfils QEIs at one-particle level would then be given by
\begin{equation}
T^{00}_\nu = (1- \nu)\,T^{00}_{\text{canonical}} +\nu\, T^{00}_{\text{alternative}} \quad \text{with}\quad |\nu|< \frac{1}{2}.
\end{equation}
We have not verified at this point whether the QEI for this energy density holds on all Fock space, however. 

Throughout this article, we have considered QEIs only on the level of one-particle states. A natural next step would be to see whether, in the same class of integrable quantum field theories, this inequality holds for all (suitably regular) states, independent of the particle number. We indeed expect that a state-independent QEI will still hold in this case; numerical evidence in two-particle states in the sinh-Gordon model suggests this. But establishing such a QEI rigorously will require substantially more technical effort. 

Specifically, let us reconsider the expansion \eqref{eq:expansion}. While in the one-particle case, only the kernel $F_2[A]$ contributes, in the general case one needs to consider all $F_k[A]$. In order to establish QEIs, we would need to compute operator norm bounds on these (or related) integral operators. However, while $F_2[A]$ is analytic, the higher-order $F_k[A]$  will possess (regularized) first-order poles in the integration variable, i.e., they are ``singular integral kernels''. Methods for finding operator bounds on these singular integral operators exist \cite{christ1990lectures}, but they are much more intricate than those we used in Sec.~\ref{sec:qei}, where we essentially estimated the kernel itself in $L^2$ norm.

Independent of that, it may also be worthwhile to generalize our results to a larger class of integrable models: While in the present paper we have considered one species of scalar Bosons, one would also like to consider models with several particle species, possibly of different masses and of higher spin, models with inner symmetries (gauge symmetries) such as the non-linear $O(N)$ $\sigma$ models \cite{BabujianFoersterKarowski:2013}, and models with bound states such as the sine-Gordon model \cite{BabujianFringKarowskiZapletal:1993}. The construction of models with several particle species has been investigated in the algebraic approach to integrable systems \cite{LechnerSchuetzenhofer:2012}, which can provide a starting point for our investigation. For models with bounds states, a rigorous treatment has only partially been achieved (see \cite{CadamuroTanimoto:scalar}). Extending our analysis to these models, one might test whether the QEI \eqref{eq:qei0} persists for more general types of self-interaction, and one would be able to study the behaviour of the lower bounds on a wider range of parameters.

\section*{Acknowledgements}

We would like to thank C.~J.~Fewster for discussions and valuable suggestions.

\bibliographystyle{alpha}
\bibliography{../../integrable,suppl_arxiv}

\end{document}